\documentclass[a4paper,11pt]{article}

\usepackage{fullpage}

\usepackage{caption}
\usepackage{cite}
\usepackage{tikz}
\usepackage{tikz-qtree}
\usepackage{boxedminipage}
\usepackage{standalone}
\usepackage{fix2col}

\usepackage[vlined, ruled, commentsnumbered, linesnumbered]{algorithm2e}

\usepackage{url}
\usepackage{color}
\usepackage{listings}

\usepackage{amsmath}
\usepackage{amssymb}

\DontPrintSemicolon

\SetKwProg{Fn}{Function}{}{}

\SetKwFunction{FnFlatten}{Flatten}
\SetKwFunction{FnMinCost}{Min}
\SetKwFunction{FnRepeated}{Repeated}
\SetKwFunction{FnRepetition}{Repetition}
\SetKwFunction{FnStrc}{Strc}
\SetKwFunction{FnStrided}{Strided}
\SetKwFunction{FnTypetree}{Typetree}

\SetKwRepeat{Do}{do}{while}

\SetKwRepeat{Struct}{struct \{}{\}}

\newcommand{\Int}{\KwSty{int}}
\newcommand{\Typenode}{\KwSty{Typenode}}
\newcommand{\Enum}{\KwSty{enum}}

\newtheorem{definition}{Definition}
\newtheorem{theorem}{Theorem}
\newtheorem{proposition}{Proposition}
\newtheorem{lemma}{Lemma}
\newtheorem{corollary}{Corollary}

\newenvironment{proof}{\par\noindent{\textsc{Proof:}}\space}{$\Box$\protect\\ \par}

\newcommand{\leaf}{\mathsf{con}}
\newcommand{\vct}{\mathsf{vec}}
\newcommand{\idx}{\mathsf{idx}}

\newcommand{\strc}{\mathsf{strc}}
\newcommand{\idxbuc}{\mathsf{idxbuc}}
\newcommand{\vctbuc}{\mathsf{vecbuc}}

\newcommand{\lookup}{\mathrm{lookup}}

\newcommand{\sequ}[1]{\langle #1 \rangle}
\newcommand{\cost}{\mathrm{cost}}

\SetAlgorithmName{Listing}{listing}{List of Listings}

\title{Polynomial-time Construction of Optimal Tree-structured
  Communication Data Layout Descriptions\thanks{This work was
    co-funded by the European Commission through the EPiGRAM project
    (grant agreement no.\ 610598).}}

\author{Robert Ganian\\
       Algorithms and Complexity Group\\
       Vienna University of Technology\\
       Austria\\
       \texttt{rganian@gmail.com}
\and Martin Kalany\\
       Parallel Computing Group\\
       Vienna University of Technology\\
       Austria\\
       \texttt{kalany@par.tuwien.ac.at}
\and Stefan Szeider\\
       Algorithms and Complexity Group\\
       Vienna University of Technology\\
       Austria\\
       \texttt{stefan@szeider.net}
\and Jesper Larsson Tr\"aff\\
       Parallel Computing Group\\
       Vienna University of Technology\\
       Austria\\
       \texttt{traff@par.tuwien.ac.at}
}

\begin{document}
\maketitle

\begin{abstract}
We show that the problem of constructing \emph{tree-structured
  descriptions of data layouts} that are optimal with respect to space
or other criteria, from given sequences of displacements, can be
solved in \emph{polynomial time}. The problem is relevant for
efficient compiler and library support for communication of
non-con\-ti\-gu\-ous data, where tree-structured descriptions with
low-degree nodes and small index arrays are beneficial for the
communication soft- and hardware. An important example is the
Message-Passing Interface (MPI) which has a mechanism for describing
arbitrary data layouts as trees using a set of increasingly general
constructors. Our algorithm shows that the so-called MPI
\emph{datatype reconstruction problem by trees} with the full set of
MPI constructors can be solved optimally in polynomial time, refuting
previous conjectures that the problem is NP-hard. Our algorithm can
handle further, natural constructors, currently not found in MPI.

Our algorithm is based on dynamic programming, and requires the
solution of a series of shortest path problems on an incrementally
built, directed, acyclic graph. The algorithm runs in $O(n^4)$ time
steps and requires $O(n^2)$ space for input displacement sequences of
length $n$.
\end{abstract}

\section{Introduction}
\label{sec:introduction}
It is a common situation for instance in parallel, numerical libraries
that substructures of large, static data structures have to be
communicated among
processors~\cite{ChoiDongarraOstrouchovPetitetWalkerWhaley96,PoulsonMarkerHammondRomerovandeGeijn13},
e.g., row- or column vectors or sub-matrices of multi-dimensional
matrices, or irregular substructures corresponding to the non-zeros or
other special elements of larger structures. This requires efficient
access to the typically non-contiguously stored substructure elements
in some predefined order, either for the application which
``(un)packs'' the elements (from) to some structured communication
buffer, or for the communication soft- or hardware to handle the
non-consecutive communication in a way that is transparent to the
application. For the latter approach, concise and efficient
descriptions of such substructures are needed. For instance, lists of
element addresses or displacements are neither concise (space
proportional to the number of elements is required) nor efficient
(processing time is at least doubled, since also the list has to be
traversed). For substructures with some regularities, much better
representations are obviously possible. Often, tree representations
are used with leaves describing base-types and interior constructor
nodes how subtrees are repeated. For example, complex data types in
C-like languages can be built recursively using a small number of
constructors (like arrays and \texttt{struct}s) from given primitive
types (\texttt{int}s, \texttt{char}s, \texttt{double}s, etc.), and the
resulting type trees describe to the compiler how data are laid out in
memory. The same kind of mechanism could be used to describe
substructures of such data types (but is not a part of C). The
Message-Passing Interface (MPI)~\cite{MPI-3.0} is an important example
of a parallel communication interface, indeed often used to implement
parallel numerical
libraries~\cite{ChoiDongarraOstrouchovPetitetWalkerWhaley96,PoulsonMarkerHammondRomerovandeGeijn13},
which provides a generic, explicit mechanism for describing
non-consecutive application data to allow the library implementation
to perform non-consecutive communication in an efficient way, possibly
by directly exploiting hardware features for, e.g., strided,
non-consecutive communication. Given such a tree-structured
description of an application data layout, it is a natural question to
ask whether this description is optimal under some given cost model
reflecting the cost of storing or processing the
description. Likewise, given a trivial description of a data layout in
the form of a long list of addresses (or offsets, or displacements),
it is natural to ask for an algorithm for constructing an efficient,
that is, cost-optimal representation as a tree with some given set of
constructors. In the MPI community, the former problem is referred to
as \emph{type normalization}, and the latter as \emph{type
  reconstruction}~\cite{Traff11:typeguide}. Both problems are
eventually important for the implementation of very high-quality MPI
libraries. The problems would be similarly important in other parallel
interfaces or languages supporting communication of arbitrarily
structured, non-consecutive data. Ideally, a compiler would be able to
perform the normalization (optimization) of data layout descriptions
given more or less explicitly by the application programmer in the
code with the constructs available in the parallel
language~\cite{SchneiderKjolstadHoefler13}.

In this paper, we investigate primarily the type reconstruction
problem for a given set of constructors, that is, the problem of
finding the most concise tree representation of a given substructure
specified by an explicit list of displacements. As the set of
constructors, we use a convenient abstraction of the type constructors
found in MPI~\cite[Chapter 4]{MPI-3.0}. This is both a natural and
powerful set that includes constructors for the case where a single
substructure is repeated in a regular or irregular pattern as well as
the case where different substructures are concatenated with given
displacements. Our main result is to show that an optimally concise
tree representation can be found in polynomial time for the whole set
of constructors, and thus as a corollary that both type reconstruction
and type normalization for the whole set of MPI derived data type
constructors can be solved in polynomial time.  This is an interesting
result since the computational hardness of the problem was not known
before.  Indeed, the problem was believed not to be in $P$ by parts of
the MPI community. Specifically, we give an algorithm that finds an
optimal type tree description for a sequence of displacements of
length $n$ in $O(n^4)$ operations. The algorithm is based on a
non-trivial use of dynamic programming requiring the solution of a
single-source shortest path problem for each new subproblem
solution. Using standard dynamic programming techniques, the space
requirement is $O(n^2)$.

MPI libraries typically employ simple forms of type normalization to
derived data types set up by the application programmer (this is
folklore, but
see~\cite{KjolstadHoeflerSnir11,KjolstadHoeflerSnir12,RossMillerGropp03}
for explicit descriptions). In recent
papers~\cite{Traff14:normalization,Traff15:mpilinear}, the problem was
more systematically analyzed, and it was shown that when restricted to
certain homogeneous constructors (those having a single child) the
reconstruction and normalization problems can be solved quite
efficiently in low, polynomial time. It was explicitly conjectured 
that the problems with the full set of MPI derived data type 
constructors would be NP-hard~\cite{Traff15:mpilinear, Traff11:typeguide}. 
We stress that when it is allowed to fold the constructed
trees into even more concise, directed acyclic graphs (DAGs), the
optimality of our construction is no longer guaranteed. We discuss
this problem at the end of the paper.

The notion of an optimal tree-like representation of a data layout is
of course relative to the way the tree will be used and processed by
the parallel programming language or library
implementation. Processing typically includes the ability to pack and
unpack parts of the layout independently using hardware support for
blocked, strided memory access and similar features of the
communication subsystem. We do not deal with the problem of efficient
datatype-tree processing here, but abstract storage and processing
costs with a simple, parameterized cost model, which must be adapted
to the concrete situation. The literature on optimization of the
processing of tree representations of data layouts in MPI is large;
some pointers are given in~\cite{Traff14:normalization}.

The paper is structured as follows. We define the set of considered
constructors and precisely formulate the type reconstruction problem
in Section~\ref{sec:problem}.  Our main result is given in
Section~\ref{sec:treereconstruc}, which describes our dynamic
programming algorithm, proves correctness and establishes the
complexity bound. In Section~\ref{sec:generalizations} we discuss how
our approach can be extended to include other convenient and in
specific situations more concise constructors, and how the problem
changes when trees can be folded into DAGs. Concluding remarks,
including a discussion of relevant future work in this area are given
in Section~\ref{sec:conclusion}.

\section{The type reconstruction problem}
\label{sec:problem}

A \emph{data layout} is an ordered sequence of relative (integer)
displacements, each indexing a certain base data type (integer, char,
floating point number) relative to some base address. Since the
semantics of base-types will not be important for the following, we
abstract the problem to consider from here onward \emph{displacement
  sequences} which we write as $D=\sequ{d_0,d_1,\ldots,d_{n-1}}$ with
the displacements $D[i]=d_i$ being indexed from $0$ to $n-1$. We point
out that the complexity of the problems that we investigate does not
change by considering full \emph{type maps} consisting of sequences of
displacements with their associated basetype (and number of bytes
occupied), as would have to be done in a concrete implementation of
our algorithms for real libraries, although of course the structure of
the reconstructed types may look different.  A \emph{segment} of an
$n$-element displacement sequence from index $i$ to index $j$ is
denoted by $D[i,j]=\sequ{d_i,d_{i+1},\ldots,d_{j}}$, $0\leq i\leq
j<n$.  A \emph{prefix} of length $c$ is the segment $D[0,c-1]$. The
displacements of the sequence are arbitrary (non-negative, negative)
integers, and the same displacement can appear more than once
(although this will normally not be the case, and is often disallowed,
e.g., for some uses of derived data types in MPI). Thinking of
displacements as (Byte) addresses, it is clear that any application
data layout can be described by a displacement sequence. The ordering
constraint (displacement sequence, \emph{not} displacement set)
implies that data are accessed in a specific order.  This is often
important for data layouts used in communication operations.

Displacement sequences typically contain regularities and some form of
structure, since they can be thought of as arising from a specific
application, and this can be exploited to obtain more concise
descriptions.  We do this by type trees, where interior
\emph{constructor nodes} describe some ordered catenation of the
layout(s) described by the child(ren) node(s).  It is natural to ask
for an efficient, polynomial time algorithm for computing the most
concise and efficient representation for a given set of constructors
and cost model.

We consider the following set of constructors that subsume
constructors found in C-like programming languages, as well as the
derived data type constructors found in MPI:

\noindent
\begin{definition}[Basic type constructors]
\label{def:baseconstructors}
A \emph{basic tree} may be constructed from the following four 
\emph{basic constructors}:
\begin{enumerate}
\item
A \emph{leaf} $\leaf(c)$ with \emph{count} $c$ describes a
sequence of $c$ adjacent relative displacements $0,1,2,\ldots,c-1$.
\item
A \emph{(homogeneous) vector} $\vct(c,d,C)$ with \emph{count} $c$
and \emph{stride} $d$ describes the catenation of $c$ sequences $C$ at
relative displacements $0, d, 2d, \ldots, (c-1)d$.
\item
A \emph{(homogeneous) index}
$\idx(c,\sequ{i_0,i_1,\ldots,i_{c-1}},C)$ with \emph{count} $c$ and
\emph{indices} $\sequ{i_0,i_1,\ldots,i_{c-1}}$ describes the
catenation of $c$ sequences $C$ at relative displacements
$i_0,i_1,\ldots,i_{c-1}$.
\item
A \emph{heterogeneous index}, or \emph{struct},
$\strc(c,\allowbreak \sequ{i_0,i_1,\ldots,i_{c-1}},\allowbreak \sequ{C_0,C_1,\ldots,C_{c-1}})$,
with \emph{count} $c$ and \emph{indices}
$\sequ{i_0,i_1,\ldots,i_{c-1}}$ describes the catenation of $c$
sequences $C_0,C_1,\ldots,C_{c-1}$ at relative displacements
$i_0,i_1,\ldots,i_{c-1}$.
\end{enumerate}
\end{definition}

For example, the displacement sequence $\sequ{3,5,7,9,11}$ can be
described by $\idx(1,\allowbreak \sequ{3},\vct(5, \allowbreak 2,
\allowbreak \leaf(1)))$.  A more involved example is shown in
Figure~\ref{fig:typetreeExample}.  Note that any displacement sequence
$D$ of length $n$ can trivially be represented as $\idx(n, D,
\leaf(1))$.

\begin{figure}[t]
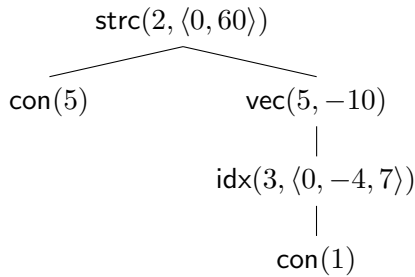

\centering
\tikzset{level distance=30pt, sibling distance=40pt}
\Tree 
    [.$\strc(2,\sequ{0,60})$ 
	[.$\leaf(5)$ ] 
	[.$\vct(5,-10)$
	    [.$\idx(3,\sequ{0,-4,7})$ 
		[.$\leaf(1)$ ]
	    ]
	]
    ]
\caption{Type tree representing the displacement sequence $D = \langle 0$, $1$, $2$, $3$, $4$, $60$, $56$, $67$, $50$, $46$, $57$, $40$, $36$, $47$, $30$, $26$, $37$, $20$, $16$, $27\rangle$.
Note that if the $\strc$ constructor is not allowed, the only way to represent this displacement sequence is the trivial representation $\idx(20, D, \leaf(1))$.}
\label{fig:typetreeExample}
\end{figure}
We refer to vertices of type trees as \emph{nodes}, where each node is
one of the constructors.

It can easily be shown that each of the MPI derived data type
constructors (for contiguous, vector, index, and structured
subtrees)~\cite[Chapter 4]{MPI-3.0} is expressible by the basic
constructors of Definition~\ref{def:baseconstructors}, and that the
mapping is almost one-to-one. For instance, the
\texttt{MPI\_Type\_vector} constructor denotes a layout consisting of
a strided sequence of blocks, each being a strided sequence of some
type $B$. This is expressed as $\vct(c,s,\vct(b,e,B))$ where $c$ is
the number of blocks, $s$ their stride, $b$ the number of elements in
each block, and $e$ the stride used within each block. We treat base
types as sequences of bytes which can be expressed by leaf nodes,
e.g., a 32-bit entity like \texttt{int} would be expressed by
$\leaf(4)$.  The $\idx$ constructor makes it possible to express the
repetition of the same layout $B$ each at some arbitrary displacement;
for this only the sequence of start indices (and the size of this
sequence) needs to be represented. The most expressive, arbitrary
\emph{branching constructor} $\strc$ can express the catenation of a
sequence of possibly different, smaller layouts each starting at an
arbitrary displacement. This is the only constructor node with arity
greater than one. In contrast to the similar MPI constructor
\texttt{MPI\_\-Type\_\-create\_\-struct}, which also takes a
repetition count (blocklength) for each substructure, the $\strc$
constructor saves this extra sequence. If a substructure is indeed a
repetition of some even smaller substructure, this information is part
of the substructure and not of the $\strc$ node itself. The basic
constructors increase in generality and storage cost: an $\idx$ node
is a $\strc$ node where all substructures are similar, and therefore
does not need to store a sequence of subtypes; a $\vct$ node is an
$\idx$ node with regularly strided displacements, which can be
computed from a single scalar instead of storing an explicit index
sequence. As the example in Figure~\ref{fig:typetreeExample} shows,
the $\strc$ constructor makes unbounded compression possible over the
$\idx$ constructor.

To make it possible to express further common patterns without
redundancy, we also consider a few auxiliary constructors. The
patterns that these constructors capture can all be expressed by
two-level nestings of basic constructors, but possibly at a higher
cost. For practical purposes and depending on the application usage
patters that are intended to be supported, it might therefore make
sense to have a richer set of constructors. For instance, MPI has both
an \texttt{MPI\_\-Type\_\-create\_\-indexed\_\-block} (which is
captured by the $\idx$ basic constructor node) and an
\texttt{MPI\_\-Type\_\-indexed} constructor which stores also a
repetition count for each index. In cases where all substructures are
repeated the same number of times, this is strictly redundant, and
there are therefore use cases for both constructors. We include the
auxiliary constructors to argue informally that our algorithm can
handle a large set of reasonable constructors.

\noindent
\begin{definition}[Auxiliary type constructors]
\label{def:auxconstructors}
An \emph{extended tree} may contain also the following two 
\emph{auxiliary constructors}:
\begin{enumerate}
\item
A \emph{strided bucket}, $\vctbuc(c,\allowbreak d,\allowbreak e,\allowbreak
\sequ{b_0,b_1,\ldots,b_{c-1}},\allowbreak C)$ with \emph{count} $c$ and
\emph{strides} $d,e$ describes the catenation of $c$ sequences at
relative displacements $0, d, 2d,\ldots (c-1)d$. The $i$-th sequence is
the catenation of $b_i$ sequences $C$ at relative displacements
$0,e,2e,\ldots (b_i-1)e$.
\item
An \emph{indexed bucket}, 
$\idxbuc(c,\allowbreak e, \allowbreak \sequ{i_0,i_1,\ldots i_{c-1}}, \allowbreak \sequ{b_0,b_1,\ldots b_{c-1}},C)$, with \emph{count} $c$
and \emph{substride} $e$ describes the catenation of $c$ sequences
at relative indices $i_0,i_1,\ldots i_{c-1}$. The $i$-th sequence
is the catenation of $b_i$ sequences $C$ at relative displacements
$0,e,2e,\ldots (b_i-1)e$.
\end{enumerate}
\end{definition}

As can be seen from the discussion above, the indexed bucket
constructor corresponds to the \texttt{MPI\_\-Type\_\-indexed}
constructor. There is no MPI counterpart of the other, arguably
natural constructor. We discuss these constructors in more detail in
Section~\ref{sec:auxiliary}.

Each basic or extended tree represents one displacement sequence,
obtained by an ordered traversal of the nodes of the type tree. This
process is called \emph{flattening} and is captured by the algorithm
in Listing~\ref{alg:flattening} for the basic constructors; the
auxiliary constructors can be handled similarly. The converse is not
true: a displacement sequence will almost always have several possible
type tree representations.

\begin{algorithm}[t]
\Fn{\FnFlatten{T, base}}{
    \Switch{T.nodetype}{
        \Case(\tcc*[f]{leaf of consecutive indices}){$\leaf$}{ 
            \For{$i \gets 0$; $i<T.c$; $i{+}{+}$}{
                print $base+i$
            }
        }
        \Case(\tcc*[f]{strided layout}){$\vct$}{
             \For{$i \gets 0$; $i<T.c$; $i{+}{+}$}{
                \FnFlatten{$T.subtype$, $base+i\cdot T.d$}
             }
        }
        \Case(\tcc*[f]{indexed layout}){$\idx$}{
             \For{$i \gets 0$; $i<T.c$; $i{+}{+}$}{
                \FnFlatten{$T.subtype$, $base+T.D[i]$}
             }
        }
        \Case(\tcc*[f]{indexed layout with subtypes}){$\strc$}{
            \For{$i \gets 0$; $i<T.c$; $i{+}{+}$}{
                 \FnFlatten($T.subtypes[i]$, $base+T.D[i]$)
            }
        }
    }
}
\caption{Flattening procedure defining the displacement sequence
  represented by a given basic tree $T$. The procedure is called
  with a base offset, which will normally be 0. The procedure can
  trivially be extended to also cover extended trees.}
\label{alg:flattening}
\end{algorithm}

We make no claim that Listing~\ref{alg:flattening} depicts a
particularly good way of implementing
flattening~\cite{Traff99:flattening}. Note that the size of the
displacement sequence described by a type tree $T$ could be much
larger than the number of nodes in $T$. Within this paper, we assume
that all numbers can be represented by a constant number of bits;
otherwise, our main result still holds, but the upper bound on space
requirements increases by a logarithmic factor.

By the \emph{conciseness} of a type tree we mean the space taken by
the representation. This is constant for vector and leaf nodes and
proportional to the size of the index and type sequences for the other
constructors. Processing costs are related to conciseness: the concise
vector constructor that describes a strided repetition of a
sub-pattern can often be handled by strided memory-copy or strided
communication operations, whereas constructors with sequences of
displacements or types need at least a traversal of the corresponding
sequences and typically entails a more irregular and expensive access
to memory. We will therefore first focus on a simple cost model for
optimizing conciseness.

The \emph{cost} of a type node shall be proportional to the number of
words that must be stored to process the node.  This includes the node
type ($\leaf, \vct, \idx, \strc$), count, displacement or pointer to
index or type array, pointer to child node(s), and a lookup cost for
the elements in lists of indices or types:
\begin{eqnarray*}
\cost(\leaf(c)) & = &K_{\leaf} \\
\cost(\vct(c,d,C)) & = & K_{\vct} \\
\cost(\idx(c,\sequ{\ldots},C)) & = & K_{\idx}+cK_{\lookup}\\
\cost(\strc(c,\sequ{\ldots},\sequ{\ldots})) & = & K_{\strc}+2cK_{\lookup}\\
\end{eqnarray*}
The constants can be adjusted to reflect other overheads related to
representing and processing a node.  We define the \emph{cost} of a
type tree $T$ to be the \emph{additive cost} of its nodes $T_i$:
$\cost(T) = \sum_i \cost(T_i)$.

\begin{algorithm}
\Struct{\Typenode}{
    \Enum{} $nodetype$ = \{$\leaf$, $\vct$, $\idx$, $\strc$\}\;
    \Int{} $c$ \tcc*{count}
    \Int{} $d$ \tcc*{stride}
    \Int{} $D[ ]$ \tcc*{displacement of subtypes}
    \Typenode{} $subtype$ \tcc*{subtype}
    \Typenode{} $subtypes$[ ] \tcc*{array of subtypes}
}
\caption{A possible \Typenode{} structure for representing nodes in
  type trees or DAGs.}
\label{lst:typenode}
\end{algorithm}

For the examples given in this paper, we take $K_{\leaf} = K_{\vct} =
K_{\idx} = K_{\strc}$, and $K_{\lookup}=1$.  For instance, with a
C-style structure as shown in Listing~\ref{lst:typenode} to represent
any of the type constructors, all constructors indeed have the same
constant in the cost (which we could take as 6 units). We remark that
our algorithm is not dependent on the specific choice of the cost
function, and that our results also hold for other reasonable cost
functions where the cost of a node is a function of the node itself
and the costs of its children.

We can now formally define the problem that we will solve in the next
section. Recall that a type tree $T$ \emph{represents} a displacement
sequence $D$ if $\texttt{Flatten}(T,0) = D$.
\begin{center}
  \begin{boxedminipage}[t]{\textwidth}
  \begin{quote}
  \textsc{Basic Type Reconstruction Problem}\\ \nopagebreak
  \emph{Instance}: A displacement sequence $D$ of length $n$.\\ \nopagebreak
  \emph{Task}: Find a least-cost (or optimal) basic tree $T$ representing~$D$; that is, $\cost(T)\leq\cost(T')$ for any basic tree $T'$ representing $D$.
\end{quote}
\end{boxedminipage}
\end{center}
\medskip

\section{Basic tree reconstruction in polynomial time}
\label{sec:treereconstruc}

We now present our main result, namely that the \textsc{Basic Type
  Reconstruction Problem} can be solved in polynomial time. 
subsequently show that extending the set of the auxiliary constructors
of Definition~\ref{def:auxconstructors}.

\begin{theorem}
\label{thm:polytree}
For any input displacement sequence $D$ of length $n$, the
\textsc{Basic Type Reconstruction Problem} can be solved in $O(n^4)$
time and $O(n^2)$ space.
\end{theorem}

\textsc{Proof outline:} We first give a characterization of the
structure of optimal basic trees (Lemma~\ref{lemma:niceTypeTree})
which allows for a simple and elegant procedure to solve the special
case of displacement sequences in \emph{normal form}
(Definition~\ref{def:normalForm}).

The fundamental observation for the proof is that any (non-trivial)
displacement sequence can be described by either a catenation of the
same kind of shorter displacement sequences (and thus by either a
vector or an index constructor) or by a catenation of different, but
shorter displacement sequences (and thus by a struct constructor). In
both cases, for an optimal description, the description of the shorter
sequences must likewise be optimal, and the principle of optimality
applies. This intuition is formalized in Lemma~\ref{lemma:repetition}
and Lemma~\ref{lemma:strc}. Lemma~\ref{lemma:optimalForNF} proves the
claim for the special case of displacement sequences in normal form,
with a detailed procedure given in Listing~\ref{alg:typetree}.

Finally, Lemma~\ref{lemma:optimalForAny} shows how to construct an
optimal basic tree for any displacement sequence out of an optimal
basic tree representation of its normal form.

\begin{definition}[Repetition, Strided Repetition]
A \emph{repetition} in a displacement sequence $D$ of length $n$ is a
prefix $C = D[0,q-1]$ of length $q$ s.t.\ $q$ is a divisor of $n$ and
for all $i$,$j$, $1\leq i <n/q$, $0\leq j<q$ we have that $D[j] - D[0]
= D[iq+j] - D[iq]$.  A \emph{strided repetition} of length $q$
additionally fulfills $D[(i+1)q] - D[iq] = D[q] - D[0]$ for all $i$,
$0\leq i < n/q - 1$, where $d = D[q] - D[0]$ is the \emph{stride} of
the repetition.
\end{definition}

\begin{algorithm}[t]
\Fn{\FnRepeated{D, n, q}}{
    \For{$i\gets q$; $i < n$; $i \gets i + q$}{
        \For{$j\gets 1$; $j<q$; $j\gets j + 1$}{
            \If{$D[j]-D[0] \neq D[i+j]-D[i]$}{
                \Return false
            }
        }
    }
    \Return true
}

\Fn{\FnStrided{D, n}}{
    $d \gets D[1] - D[0]$\;
    \For{$i \gets 1$; $i < n$; $i\gets i + 1$} {
        \lIf{$D[i] - D[i - 1] \neq d$} {\Return false}
    }
    \Return true
}
\caption{Trivial checks for repetitions and strided repetitions.}
\label{lst:subroutines}
\end{algorithm}

The intention of the functions \FnRepeated and \FnStrided (see
Listing~\ref{lst:subroutines}) is to find (strided) repetitions $C$ of
a displacement sequence $D$ that can be exploited to represent $D$ via
an $\idx$ or $\vct$ constructor with subsequence $C$.  It is easy to
see that \FnRepeated and \FnStrided as outlined both take linear time.

As mentioned above, any displacement sequence $D$ can be described by
either a catenation of the same kind of shorter displacement sequences
or by a catenation of different, but shorter displacement sequences.
Additionally, a representation via a $\leaf$ node is possible if $D$
is a trivial displacement sequence $\sequ{0,1,\dots,\allowbreak n-1}$.
In terms of type trees, this means that an optimal basic tree $T$ for
a displacement sequence $D$ is either
\begin{enumerate}
\item $T = \leaf(n)$, a single $\leaf$ node with count $n$; or
\item $T = \vct(c, d, S)$, where the prefix $D[0,q-1]$ of length $q =
  n/c$ is a strided repetition in $D$ with stride $d$ and $S$ is an
  optimal basic tree for the prefix $\sequ{D[0],\dots, D[q-1]}$; or
\item $T = \idx(c, \sequ{i_0,\dots,i_{c-1}},S)$, where the prefix
  $D[0,q-1]$ of length $q = n/c$ is a repetition in $D$, $S$ is an
  optimal basic tree for the sequence $\sequ{D[0]-i_0,\dots,
    D[q-1]-i_0}$ and the indices $i_0,\dots,i_{c-1}$ are such that
  $\FnFlatten(T,0) = D$; or
\item $T = \strc(c, \sequ{i_0,\dots,i_{c-1}}, \sequ{S_0,\dots,
  S_{c-1}})$, where the $S_j$ for $0\leq j < c$ are optimal basic
  trees for some sequences $C_j$ which together with the indices
  $i_0,\dots,i_{c-1}$ are such that $\FnFlatten(T,0) = D$.
\end{enumerate}
While the first case can be handled with a single scan of $D$, the
others are more involved.  In the following, we give a more detailed
characterization of (optimal) basic trees to tackle the problem.

\begin{definition}[Shifted node]
We call an index node $\idx(c,\allowbreak \sequ{i_0,\dots},\allowbreak
C)$ or a struct node $\strc(c, \allowbreak
\sequ{i_0,\dots},\allowbreak \sequ{\dots})$ with $i_0 \neq 0$ a
\emph{shifted node}; $s= i_0$ is called the node's \emph{shift}.
\end{definition}
Note that adding some value $s$ to all indices of an $\idx$ or $\strc$
node $N$ shifts the sequence represented by the basic tree rooted at
$N$ by $s$.

\begin{definition}[Nice basic tree]
\label{def:nicetree}
A \emph{nice basic tree} contains at most one shifted node, which is the first $\idx$ or $\strc$ node on every root to leaf path.
\end{definition}

\begin{lemma}
\label{lemma:niceTypeTree}
For any basic tree $T$ representing a displacement sequence $D$, a
nice basic tree representation $\tilde{T}$ of $D$ of equal cost
exists.
\end{lemma}
\begin{proof}
A node is \emph{bad} if it is a shifted node and it is not the first
$\idx$ or $\strc$ node on every root to leaf path.  Let $D$ be a fixed
displacement sequence and let $T$ be a basic tree representing $D$
with a minimum number of bad nodes. We will show that $T$ is, in fact,
nice.

Assume that a bad index node (the proof is analogous for a bad struct
node) $N_I = \idx(c, \allowbreak \sequ{i_0,\dots,\allowbreak
  i_{c-1}},\dots)$ is present in the $k$-th subtree of a struct node
$N_S = \strc(c',\allowbreak \sequ{{i'_0},\dots,\allowbreak i'_k,\dots,
  \allowbreak i'_{c'-1}},\sequ{\dots})$ s.t.\ there is no other
shifted node on the path from $N_I$ to $N_S$.  We can change $N_I$ to
a non-shifted index node by subtracting its shift $s=i_0$ from all
indices $i_j$, for $0\leq j < c$ and adding $s$ to the $k$-th index
$i'_k$ of $N_S$, i.e., $\tilde{N_I} = \idx(c,\allowbreak
\sequ{0,i_1-s,\dots,\allowbreak i_{c-1}-s},\dots)$ and $\tilde{N_S} =
\strc(c',\allowbreak \sequ{{i'_0},\dots,\allowbreak i'_k+s,\dots ,
  i'_{c'-1}},\sequ{\dots})$. Notice that the basic tree obtained in
this way still represents the same displacement sequence $D$ but
contains one less bad node, and hence the existence of such a node
$N_I$ would contradict our choice of $T$.
    
Hence there is no $\strc$ node on the path from a bad node $N_I$ to
the root node $R$.  If this path contains an index node $N'_I \neq
N_I$, proceed analogously to the previous case: $\tilde{N_I} =
\idx(c,\allowbreak \sequ{0,i_1-s,\dots,\allowbreak i_{c-1}-s},\dots)$
and $\tilde{N'_I} = \idx(c',\allowbreak \sequ{i'_0+s,\dots,\allowbreak
  i'_{c'-1}+s},\dots)$. Again, the obtained basic tree also represents
$D$ but contains one less bad node, contradicting our original choice
of $T$. Consequently, $T$ does not contain any bad nodes and thus must
be a nice basic tree.
\end{proof}

\begin{corollary}
\label{corollary:atMostOneIndexWithCount1}
Any optimal basic tree $T$ contains at most one index node with count
1, i.e., at most one node of the form $N =
\idx(1,\sequ{i_0},\dots)$. Additionally, there is no other $\idx$ or
$\strc$ node on the path from $N$ to the root.
\end{corollary}
\begin{proof}
Assume that $T$ contains two index nodes with count 1.  Since $T$ is a
tree, there is an index node $N$ with count 1 s.t.\ the path from $N$
to the root node of $T$ contains another $\idx$ or $\strc$ node.  In a
cost-equivalent nice basic tree representation $\tilde{T}$ (obtained
by applying the procedure from the proof of
Lemma~\ref{lemma:niceTypeTree}), the corresponding index node is
$\tilde{N} = \idx(1,\sequ{0},T')$.  Note that the type tree rooted at
$\tilde{N}$ represents exactly the same displacement sequence as its
subtype $T'$. Thus a representation $T'$ of less cost exists, which
contradicts the assumption that $T$ is optimal.
\end{proof}

The following proposition, although not directly required for the
analysis, provides some additional insight into the structure of
optimal basic trees.

\begin{proposition}
\label{proposition:heightOfOptimalTypeTree}
The height of an optimal basic tree is $O(\log n)$.
\end{proposition}
\begin{proof}
It is easy to see that an optimal basic tree does not contain two
consecutive $\strc$ nodes, as they can always be merged into one while
reducing the cost.  For any basic tree $T$ that represents a sequence
of length $n$, a basic tree $\idx(c,\sequ{\dots},T)$ or
$\vct(c,\dots,T)$ with $c\geq 2$ represents a sequence of length at
least $2n$. Let $P$ be a maximum-length path from a leaf to the root
of an arbitrary optimal basic tree.  Since any optimal basic tree
contains at most one $\idx$ node with count $c=1$
(Corollary~\ref{corollary:atMostOneIndexWithCount1}) and no $\vct$
node with $c=1$, the length of the represented sequence at least
doubles with at least every other node on $P$.
\end{proof}

\begin{definition}
\label{def:normalForm}
The \emph{normal form} $\hat{D}$ of a displacement sequence $D$ of
length $n$ is defined as $\hat{D}[i] = D[i] - D[0]$, for all $i$,
$0\leq i < n$.
\end{definition}
In other words, the normal form $\hat{D}$ of a displacement sequence
$D$ is obtained by shifting $D$ so that its first element is $0$.

\begin{corollary}
\label{corollary:structureOfNFSolutions}
An optimal basic tree $T$ for a displacement sequence $\hat{D}$ in
normal form does not contain any shifted nodes or any $\idx$, $\vct$
or $\strc$ node with count 1.
\end{corollary}
\begin{proof}
It follows directly from Lemma~\ref{lemma:niceTypeTree} and
Corollary~\ref{corollary:atMostOneIndexWithCount1} that there exists
an optimal basic tree $T$ for $\hat{D}$ which does not contain any
shifted nodes.  Note that a non-shifted $\idx$, $\vct$ or $\strc$ node
with count 1 does not change the represented sequence.  Thus, removing
such nodes from a basic tree reduces the cost while not changing the
represented displacement sequence.  It follows that no such node can
be part of an optimal basic tree.
\end{proof}

Observe that since there are no shifted nodes in an optimal basic tree
$T$ for $\hat{D}$, any subtree of $T$ represents a segment of
$\hat{D}$ in normal form.  In the following, we will use $T_{i,j}$ to
denote an optimal basic tree representation for the normalized segment
$\hat{D}[i,j]$ of $\hat{D}$.

For convenience, we define the function $\FnMinCost(S, T)$ which,
given two basic trees $S$ and $T$, returns the one with least cost (if
either is \verb|null|, the other is returned).  Note that the cost of
a basic tree can trivially be computed by a simple traversal.
However, when constructing basic trees from the bottom up (as we will
do in this section), we keep for each node the cost of the subtree
rooted at that node. This allows for the cost of a basic tree to be
queried in constant time and thus for a constant-time implementation
of $\FnMinCost$.

\begin{algorithm}
\Fn{\FnRepetition{$D$, $n$}}{
    $T_r \gets $ null\;
    \ForEach{\text{divisor} $q$ of $n$, $q<n$}{
        $c \gets n/q$\;
        \If{\FnRepeated{$D$, $n$, $q$}}{
            \For{$i=0$; $i < c$; $i{+}{+}$}{ 
                $I[i] \gets D[iq]$\;
            }
            $T_{idx} \gets \idx(c, I, T_{0,q-1})$\;
            $T_r$ = \FnMinCost{$T_{idx}$, $T_r$}\;
            \If{\FnStrided{$I$, $c$}} {
                $d \gets I[1] - I[0]$\;
                $T_{vec} \gets \vct(c,d,T_{0,q-1})$\;
                $T_r \gets $ \FnMinCost{$T_{vec}$, $T_r$}\;
            } 
        }
    }
    \Return $T_r$
}
\caption{Algorithm to find a least-cost representation for a
  displacement sequence in normal form with an $\idx$ or $\vct$ node
  as root node.}
\label{alg:checkForRepetitions}
\end{algorithm}

\begin{lemma}
\label{lemma:repetition}
Let $\hat{D}$ be any displacement sequence of length $n$ in normal
form and assume that optimal basic tree representations for all normal
form prefixes of length less than or equal to $\lfloor n/2\rfloor$ are
known.  A representation $T_{r}$, where the root node of $T_{r}$ is
either an $\idx$ or a $\vct$ node and $T_r$ is of least cost
w.r.t.\ all possible representations of that form, can be computed in
$O(n\sqrt{n})$ time.
\end{lemma}
\begin{proof}
Listing~\ref{alg:checkForRepetitions} enumerates all possible
representations of the desired form and chooses the one with least
cost among them.  Note that for the divisor $q=1$, the trivial
representation $\idx(n, \hat{D}, \leaf(1))$ (which exists for any
displacement sequence $\hat{D}$), is generated and thus a valid
representation for $\hat{D}$ is guaranteed to be found.  For the same
reasons as given in Corollary~\ref{corollary:structureOfNFSolutions},
$\idx$ nodes with count 1 cannot be part of a least-cost
representation of the desired form and thus need not be considered.

The number of divisors of $n$ is upper-bounded by
$2\lfloor\sqrt{n}\rfloor$ and, by assumption, optimal representations
for all prefixes of $\hat{D}$ of length less than or equal to $\lfloor
n/2 \rfloor$ are known, i.e., $T_{0,j}$ is known for all $j$, $O \leq
j \leq n/2$.  This implies the claimed runtime bound.
\end{proof}

\begin{lemma}
\label{lemma:strc}
Let $\hat{D}$ be any displacement sequence of length $n$ in normal
form and assume that optimal basic tree representations are known for
all normal form segments of length strictly less than $n$.  A
representation $T_{s}$, where the root node of $T_{s}$ is a $\strc$
node and $T_s$ is of least cost w.r.t.\ to all possible
representations of that form, can be computed in $O(n^2)$ time.
\end{lemma}
\begin{proof}
Construct a weighted, directed acyclic graph $G=(V,E,w)$ with $V =
\{v_0,\dots,v_{n}\}$, $E = \{(v_i,v_j) \mid 0 \leq i < j \leq n,\; j -
i < n\}$ and the weight function $w$ which is defined for all edges
$(v_i,v_j)$ in $E$ as $w(v_i,v_j) = 2K_{\lookup} + \cost(T_{i,j-1})$.
The intended meaning of this construction is as follows.  A node $v_i$
corresponds to the $i$-th element of $\hat{D}$ ($v_n$ is a special
vertex that corresponds to the hypothetical first element after the
end of $\hat{D}$) and an edge $(v_i,v_j)$ with $i<j$ corresponds to
the segment $\hat{D}[i,j-1]$ in normal form.  The weight of an edge
$(v_i,v_j)$ is equal to the cost of the optimal representation
$T_{i,j-1}$ of the segment $\hat{D}[i,j-1]$ (which exists by the
assumption) plus a cost of $2K_{\lookup}$ for including this
representation as a subtype in a $\strc$ node.  The edge $(v_0, v_n)$,
which is not part of the constructed graph, can be thought of as
corresponding to the type tree $T_{0,n-1}$, i.e., the optimal type
tree representation of $\hat{D}$ we want to compute.

Let $P = \sequ{v_{0}, u_1,\dots, u_k, v_{n}}$ be a shortest path in
$G$ from $v_0$ to $v_n$ with $u_i \in V$ for $1\leq i\leq k$.  Then
the basic tree $\strc(k+1,
\sequ{\hat{D}[0],\hat{D}[u_1],\dots,\hat{D}[u_k]},
\sequ{T_{0,u_1-1},T_{u_1,u_2-1},\dots T_{u_{k},n-1}})$ is a valid
representation of $\hat{D}$.  Note that by construction, for any valid
representation of $\hat{D}$ of the desired form, a corresponding path
from $v_0$ to $v_n$ exists in $G$ and thus a shortest path represents
the desired solution of least cost.  Given $P$, this representation
can be constructed in linear time, since optimal representations for
all required segments are known by the assumption.  The resulting
graph has $n\choose 2$ edges and the runtime is dominated by the cost
of $O(n^2)$ time for finding a shortest path in a DAG.
\end{proof}

We can now give the complete dynamic programming algorithm for
constructing optimal basic trees for displacement sequences in normal
form, which proves Lemma~\ref{lemma:optimalForNF}.  Due to
Lemma~\ref{lemma:niceTypeTree}, it suffices to construct an optimal
nice basic tree which according to
Corollary~\ref{corollary:structureOfNFSolutions} cannot contain any
shifted nodes nor any $\idx$, $\vct$ or $\strc$ nodes with count 1.
The algorithm is shown in Listing~\ref{alg:typetree}.

\begin{lemma}
\label{lemma:optimalForNF}
For any input displacement sequence $\hat{D}$ of length $n$ in normal
form, the \textsc{Basic Type Reconstruction Problem} can be solved in
$O(n^4)$ time and $O(n^2)$ space.
\end{lemma}

\begin{algorithm}[ht!]
\Fn{\FnTypetree{$\hat{D}$, $n$}}{
    \tcc{Initialization}
    $G=(\{v_0,\ldots,v_n\},\emptyset)$\;
    \tcc{Preprocessing: find leaf nodes}
    \For{$i\gets 0; \; i \leq n; \; i{+}{+}$}{
        $j\gets i$\;
        \Do{$j \leq n$ and $\hat{D}[j] - \hat{D}[j-1] == 1$}{
            $T_{i,j} \gets \leaf(j-i+1)$\;
            $w_{i,j} \gets 2+cost(T_{i,j})$ \;
            Add edge $(v_i,v_{j+1})$ with basic tree $T_{i,j}$ and weight $w_{i,j}$ to $G$\;
            $j\gets j + 1$\;
        }
    }
    \tcc{Find solutions for all segments}
    \For{$l \gets 2; \; l \leq n; \; l{+}{+}$}{
        \For{$i\gets 0; \; i \leq n-l; \; i{+}{+}$}{
            \tcc{Compute optimal basic tree for normalized segment $\hat{D}[i,i+l-1]$}
            $j \gets i+l-1$\;
            \tcc{Find best representation with $\idx$ or $\vct$ node as root}
            Let $\hat{D}_{i,j}$ be the normalized segment $\hat{D}[i,j]$\;
            $T_r \gets $ \FnRepetition{$\hat{D}_{i,j}$, $l$, $i$}\; 
            $T_{i,j} \gets $ \FnMinCost{$T_r$, $T_{i,j}$}\;
            \tcc{Find best representation with $\strc$ node as root}
            Find shortest path $P$ from $v_i$ to $v_j+1$ in $G$\;
            Assume $P = \sequ{v_i, u_1, \dots u_k, v_j}$\;
            $I \gets \sequ{0,\hat{D}[u_1]-\hat{D}[v_i],\dots , \hat{D}[u_{k}]-\hat{D}[v_i]}$\;
            $subtypes \gets \sequ{T_{v_i,u_1-1}, T_{u_1,u_2-1}, \dots ,T_{u_k,v_{j-1}}}$
            $T_s \gets \strc(k+1, I, subtypes)$\;
            $T_{i,j} \gets $ \FnMinCost{$T_s$, $T_{i,j}$}\;
            Add edge $(v_i,v_{j+1})$ with representation $T_{i,j}$ and weight $K_{\lookup} + cost(T_{i,j})$ to $G$\;
        }
    }
    \Return{$T_{0,n-1}$} \tcc*{Stored with edge $(v_0, v_n)$}
}
\caption{Algorithm to find a least-cost basic tree representation.}
\label{alg:typetree}
\end{algorithm}

\begin{proof}
The input to the algorithm is an $n$-element displacement sequence
$\hat{D}$ in normal form.  The algorithm computes an optimal basic
tree $T[i,j]$ for each normalized segment $\hat{D}[i,j]$, $0\leq i\leq
j<n$, which is stored with edge $(i,j+1)$ in the constructed graph
$G$.  Note that the solution for the whole input sequence $\hat{D}$
can be read off of the edge $(v_0,v_{n})$.

The algorithm starts with a preprocessing step to find all segments
whose normal form is representable with a single $\leaf$ node.  Note
that the normal form of any segment of length 1 can trivially be
represented as $\leaf(1)$ and since no other valid representations
exist for this particular kind of displacement sequence, this
representation is optimal.  A straight forward implementation of this
preprocessing step as in Listing~\ref{alg:typetree} is clearly
feasible in time $O(n^2)$.

The algorithm computes optimal basic tree representations for all
normalized segments of $\hat{D}$, via a bottom up dynamic programming
approach.  The dynamic programming table to be filled in is implicit
in the graph $G$, where each segment $\hat{D}[i,j]$ is associated with
an edge $(v_i, v_{j+1}$).  Note that after the preprocessing step,
solutions for all segments of length 1 are known.  By incrementally
computing optimal representations for all segments of length
$2,\dots,n$, it is ensured that Lemmas~\ref{lemma:repetition}
and~\ref{lemma:strc} can be applied to compute an optimal
representation for each segment as follows.  A basic tree $T_r$, whose
root node is either an $\idx$ or a $\vct$ node, and a basic tree
$T_s$, whose root node is a $\strc$ node, are computed.  Both are of
least cost w.r.t.\ all basic tree representations of the desired form.
The optimal basic tree for a normalized segment $\hat{D}[i,i+l-1]$ is
necessarily one of $T_r$, $T_s$ or a representation via a $\leaf$ node
(if such a representation is possible), which was already computed in
the preprocessing step.

To compute $T_r$, a small, technical extension of procedure
$\FnRepetition$ (Listing~\ref{alg:checkForRepetitions}) for finding
representations via $\idx$ or $\vct$ nodes is necessary.  The
procedure requires access to optimal representations of the prefixes
of the argument displacement sequence $D$. However, in the general
case, $D$ is a segment of $\hat{D}$, that is, $D = \hat{D}[i,j]$, and
its prefixes therefore start with $\hat{D}[i]$.  To account for this
(and avoid copying $\hat{D}[i,j]$), we pass an additional argument $o$
representing the offset of the segment within the input displacement
sequence $\hat{D}$ (i.e., for a segment $\hat{D}[i,j]$, we have
$o=i$), and in lines 10 and 12 replace the argument $T_{0, q-i}$ with
$T_{o, o+q-i}$.

To compute $T_s$ in Listing~\ref{alg:typetree}, contrary to
Lemma~\ref{lemma:strc}, we do not construct a new graph for each
segment when computing its representation $T_s$.  Instead a single
dynamic, incrementally built graph $G$ suffices to solve the problem
for all segments of $\hat{D}$.  By construction, when computing the
desired representation of a segment $\hat{D}[i,i+l-1]$, $G$ contains
edges representing optimal representations for all segments of length
less than $l$ (and possibly some edges representing solutions of
length $l$).  A shortest path from node $v_i$ to $v_{i+l}$ in $G$
therefore leads to the same representation as the one constructed by
Lemma~\ref{lemma:strc}.

To find such a shortest path, for each segment $\hat{D}[i,i+l-1]$ of
length $l$, one single-source shortest path (SSSP) problem on a
weighted DAG with $l+1$ nodes and $O(l^2)$ edges has to be solved.
Since $G$ is a topologically sorted DAG by construction, SSSP is
solvable in $O(|V| + |E|)$ time, where $|V|$ denotes the number of
vertices and $|E|$ denotes the number of edges in
$G$~\cite{CormenLeisersonRivestStein09}.  To compute the desired
representations for all segments of length $l$, a shortest path has to
be computed for each of the $n+1-l$ node pairs $(v_i,v_{i+l})$, for
$0\leq i \leq n+1-l$.  The total runtime is thus upper bounded by
$\sum_{l=1}^{n+1} l^2 (n+1-l)$, which is $O(n^4)$.

The algorithm constructs a graph with $O(n^2)$ edges, where a basic
tree $T_{i,j}$, representing the solution for the normalized segment
$\hat{D}[i,j]$, is associated with each edge $(v_i, v_j)$.  Note that
for each edge $(v_i, v_j)$ it suffices to store the root node of the
associated basic tree $T_{i,j}$ plus pointers to its child nodes,
which are already stored with the respective edges.  To meet the
desired space bound, only a constant amount of space may be used by
each edge and associated basic tree.  This is trivially true for
$\leaf$ nodes (apart from one word indicating the node's kind and the
cost of the type tree rooted at the node, only the count $c$ needs to
be stored) as well as $\vct$ nodes (two integer values and one pointer
to the child node are required in addition to the node's kind and the
cost of the type tree rooted at this node).  However, $\idx$ and
$\strc$ nodes may require $\Omega(n)$ space in the worst case (e.g.,
if $\idx(n, \hat{D}, \leaf(1))$ is the optimal representation of
$\hat{D}$).  We employ a standard trick often used in dynamic
programming algorithms and store for each node only the information
required to reconstruct the full solution once the algorithm in
Listing~\ref{alg:typetree} has terminated.  If for an $\idx$ node the
count $c$ is known, the full $\idx$ node is easily derived as $\idx(c,
\sequ{\hat{D}[0], \hat{D}[q],\dots, \hat{D}[(c-1)q]}, T_{0,q-1})$ with
$q = n/c$.  The parameters of a $\strc$ node associated with an edge
$(v_i, v_j)$ can be reconstructed by again computing the shortest path
from node $v_i$ to $v_j$ and mapping it to a $\strc$ node as done in
Lemma~\ref{lemma:strc}.  Note that this reconstruction step does not
change the asymptotic runtime bound and that the required space for
each node is $O(1)$, from which the claimed upper bound of $O(n^2)$
space follows directly.
\end{proof}

The following Corollary~\ref{corollary:idxWithCount1IsRoot} and
Lemma~\ref{lemma:optimalForAny} show how the algorithm of
Lemma~\ref{lemma:optimalForNF} can be applied to general displacement
sequences.
\begin{corollary}
\label{corollary:idxWithCount1IsRoot}
For any optimal basic tree with an index node $N$ with count 1, i.e.,
a node $N = \idx(1,\sequ{i_0},\dots)$, a representation $T'$ of equal
cost s.t.\ $N$ is the root node of $T'$, exists.
\end{corollary}
\begin{proof}
Due to Corollary~\ref{corollary:atMostOneIndexWithCount1}, there is no
$\idx$ or $\strc$ node on the path from $N$ to the root and thus $N$
shifts the whole sequence by $i_0$.  This shift can be represented
equivalently by removing $N$ from the basic tree and adding a new root
node to represent the shift, i.e., by letting $T' = \idx(1,
\sequ{i_0}, T \setminus N)$.
\end{proof}

\begin{lemma}
\label{lemma:optimalForAny}
Given optimal basic trees $\hat{T}_{i,j}$ for all normalized segments
$\hat{D}[i,j]$ of a displacement sequence $D$, an optimal basic tree
$T$ representing $D$ can be computed in $O(n^2)$ time and $O(n)$
space.
\end{lemma}
\begin{proof}
By Lemma~\ref{lemma:niceTypeTree}, for any optimal basic tree $T$ a
cost-equivalent nice basic tree $\tilde{T}$ representing the same
displacement sequence $D$ exists and it therefore suffices to find an
optimal nice basic tree representation $\tilde{T}$ for $D$.  By
assumption, an optimal nice basic tree representation $\hat{T} =
\hat{T}_{0,n-1}$ for the normalized sequence $\hat{D}$ exists.  To
construct $\tilde{T}$, find the first node $N$ on any root to leaf
path in $\hat{T}$ that is either an $\idx$ or a $\strc$ node and add
the displacement sequence's shift $s = D[0]$ to the indices of this
node, i.e., if $N = \idx(c, \sequ{i_0,\dots,i_{c-1}}, \hat{T'})$ in
$\hat{T}$, set $\tilde{N} = \idx(c, \sequ{i_0 + s,\dots,i_{c-1} + s},
\hat{T'})$ in $\tilde{T}$ and analogously for the case of $N$ being a
$\strc$ node.  Note that $\tilde{T}$ has the same cost as $\hat{T}$
and thus is an optimal basic tree representation for $D$.

If such a node does not exist, it follows from
Lemma~\ref{lemma:niceTypeTree} and
Corollary~\ref{corollary:idxWithCount1IsRoot} that the optimal
solution is either
\begin{itemize}
 \item $\tilde{T} = \idx(c,\sequ{\dots}, \hat{T}_{0,n/c-1})$, for some divisor $c$ of $n$, or
 \item $\tilde{T} = \strc(c, \sequ{\dots},\sequ{\hat{T}_0,\dots, \hat{T}_{c-1}})$, for some $c$, $1< c<n$.
\end{itemize}
Note that for $\idx$ nodes, both the trivial representation $\idx(n,
D, \leaf(1))$ as well as the representation $\idx(1, \sequ{D[0]},
\tilde{T}$ which only adds a shifted node to $\tilde{T}$ need to be
checked.  Since solutions for all normalized segments are already
known, this construction is feasible in $O(n^2)$ time and $O(n)$
space.
\end{proof}

\begin{proof}[of Theorem~\ref{thm:polytree}]
The \textsc{Basic Type Reconstruction Problem} for a displacement
sequence $D$ of length $n$ can be solved by computing an optimal basic
tree representation for the normalized displacement sequence $\hat{D}$
(Lemma~\ref{lemma:optimalForNF}) and the post-processing step given in
Lemma~\ref{lemma:optimalForAny}.  The claimed space and time bounds
follow directly from the given Lemmas.
\end{proof}

\section{Computing more concise representations}
\label{sec:generalizations}

In this section we discuss possibly more space efficient tree
representations by allowing a richer set of constructors, exemplified
by the auxiliary constructors introduced in
Definition~\ref{def:auxconstructors}. We then explain why computing
representations by DAGs is an apparently harder problem. Finally, we
discuss the applicability of our algorithms to the type normalization
problem.

\subsection{Handling the auxiliary constructors}
\label{sec:auxiliary}

The auxiliary constructors of Definition~\ref{def:auxconstructors} can
be handled by slight extensions to our algorithm in a way that
polynomial-time type reconstruction is still possible. Basically, only
the part that checks for vector or index patterns shown in
Listing~\ref{alg:checkForRepetitions} needs to be extended. Assume
that a repeated prefix $C$ of length $q$ has been found in the given
displacement sequence $D$, and that $D'$ is the displacement sequence
consisting of every $q$th element of $D$, i.e.,
$D'=[D[0],D[q],D[2q],\ldots]$.

The \emph{strided bucket}, $\vctbuc(c, d, e,
\sequ{b_0,b_1,\ldots,b_{c-1}}, C)$ constructor can concisely describe
application data layouts consisting of buckets each with some maximum
number of elements (the stride $d$) where each bucket contains some
(possibly different) number of elements $b_i$ with bucket stride
$e$. This description is likely to be less costly than describing such
a layout by a $\strc$ constructor with each subtype describing one
bucket.  To incorporate the strided bucket it simply has to be checked
in Listing~\ref{alg:checkForRepetitions} whether $D'$ follows the
strided bucket pattern, and this can easily be done in linear
time. There are two cases to consider. If the first bucket has more
than one element, take as bucket stride $e=D'[1]-D'[0]$ and scan the
index list for repetitions at stride $e$. The first violation at some
position $i$ forces the maximum bucket size to be $d=D'[i]-D'[0]$. Now
continue to scan till the end of $D'$, checking that the $e,d$ strided
pattern repeats and counting the number of elements $b_i$ in each
bucket of $e$-strided displacements.  Otherwise, the first bucket has
only one element. Take instead as maximum bucket size $d=D'[1]-D'[0]$,
and scan for repetitions with stride $d$. The first violation at some
position $i$ forces the bucket stride to be $e=D'[i]-D'[i-1]$. As in
the other case, the bucket sizes $b_i$ are counted by scanning $D'$
till the end. If an index $i$ is found where $D'[i]-D'[i-1]\neq e$ and
$D'[i]-D'[j]\neq d$ where $j$ is the start of the current bucket in
$D'$, then $D'$ is not a displacement sequence of a strided bucket
layout.

The strided bucket constructor is in a sense the opposite of the index
constructor. Instead of an index sequence it takes a sequence of
bucket sizes, and has (roughly) the same cost. Interestingly, there is
no such constructor in the MPI standard.

The \emph{indexed bucket}, $\idxbuc(c,\allowbreak e,\allowbreak
\sequ{i_0,i_1,\ldots i_{c-1}},\allowbreak \sequ{b_0,b_1,\ldots
  b_{c-1}},\allowbreak C)$, on the other hand corresponds closely to
the \texttt{MPI\_\-Type\_\-indexed} constructor. For each index, a
repetition count $b_i$ gives the number of repeats of $C$ in the
bucket starting at that index; all repetitions use the same stride $e$
(the constructor could trivially be extended to the case where each
index has its own stride). For each possible bucket stride, the number
of buckets that this stride will give rise to has to be counted. The
stride $e$ leading to a smallest number of buckets is a candidate for
the representation of $D'$ and $C$ as an $\idxbuc$ node. We observe
that each $i$ with $D'[i+1]-D'[i]=e$ joins two $e$-strided segments
$D'[j,i]$ and $D'[i+1,k]$ into one bucket starting at index
$j$. Therefore, the stride that occurs most often in the stride
sequence $S[i]=D'[i+1]-D'[i]$, $0\leq i < n -1$, will lead to the
smallest number of buckets. To count the number of occurrences of each
stride, we either sort $S$ or count by hashing during the scan of
$D'$. Let $e$ be a stride with the most occurrences. A final scan of
$D'$ suffices to compute the start indices and sizes of the buckets
with stride $e$.

\subsection{Type reconstruction into DAGs}
\label{sec:dagreconstruc}

A type tree describing some given displacement sequence may have
multiple instances of the same subtree. Our algorithm in particular
constructs nice type trees (Definition~\ref{def:nicetree}) in which
all displacement sequences in index and struct nodes except perhaps
one start at index $0$, and it can well happen that the same index or
struct node occurs many times. A more concise representation results
if such trees are folded into directed acyclic graphs with only one
node for each substructure.

Type DAGs represent displacement sequences by the same flattening
procedure as shown in Listing~\ref{alg:flattening} for trees. Each
path from the root node in the type DAG to a leaf is traversed in
order to generate the corresponding displacement sequence. Thus the
processing cost of a type DAG would arguably be similar to the
processing costs of a tree. By a similar traversal of a DAG an
equivalent tree can be constructed, simply by making a new copy each
time a node is visited.

The space required for the DAG can be much smaller than the space
required for the corresponding tree. One can therefore define also for
DAGs our cost model for optimizing conciseness as the additive cost of
the nodes in the DAG; and \emph{not} as the sum of the costs of all
paths traversed. The type reconstruction problem into DAGs is now to
find the least-cost DAG representing the given displacement sequence.

One crucial difficulty which arises when dealing with such type DAGs
is that the best representation for a subsequence no longer needs to
be locally optimal, since costs savings can be achieved by reusing
other nodes of the DAG. This is illustrated in
Figure~\ref{fig:DAGExample}.

\begin{figure}[h]
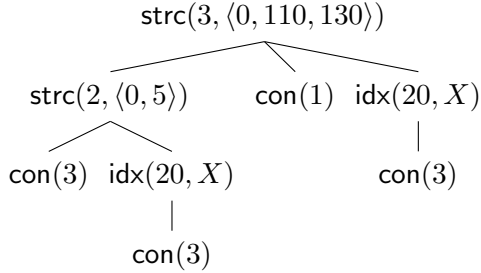

\centering
\tikzset{level distance=30pt, sibling distance=0pt}
\Tree 
    [.$\strc(3,\sequ{0,110,130})$ 
	[.$\strc(2,\sequ{0,5})$
	    [.$\leaf(3)$ ]
	    [.$\idx(20,X)$ 
		[.$\leaf(3)$ ]
	    ]
	]
	[.$\leaf(1)$ ] 
	[.$\idx(20,X)$ 
		[.$\leaf(3)$ ]
	]
    ]
\caption{An unfolding of an optimal type DAG representing a
  displacement sequence $D$; $X$ is an arbitrary subsequence of length
  $20$ over $0,1,\dots,99$. The subtrees rooted at $\idx(20,X)$ only
  contribute to the cost function once. Notice that the subtree rooted
  at $\strc(2,\sequ{0,5})$ is not a least-cost type tree
  representation of the represented subsequence.  }
\label{fig:DAGExample}
\end{figure}

In particular, this implies that the type tree constructed by
unfolding a cost-optimal DAG is not necessarily a cost-optimal tree,
and conversely, that the DAG obtained by folding a given, cost-optimal
type tree is not necessarily a cost-optimal DAG.  This constitutes a
fundamental problem for our general approach for handling type trees,
and new ideas are needed to solve the type reconstruction problem into
DAGs.

\subsection{The type normalization problem}

The type normalization problem subsumes the type reconstruction
problem that we have considered so far. Type normalization asks to
improve the cost of an already given tree description of the data
layout. Since any data layout can be represented as a single $\idx$
node with the whole displacement sequence as index sequence, type
normalization includes type reconstruction as a special case. Type
normalization is the problem that compiler or library implementors are
typically faced with: application data structures described as trees
are given by the programmer as part of the code, and an internal,
optimal representation is to be constructed by the programming system.
 
The trivial solution is to flatten the given type tree and apply the
type reconstruction algorithm on the resulting displacement
sequence. Since the size of the resulting displacement sequence is not
bounded by the size or conciseness of the tree, this is highly
undesirable. We would like a procedure where the complexity can be
bounded by the conciseness of the type trees, specifically the total
size of the index sequences in the tree.

As shown in~\cite{Traff14:normalization}, if the set of basic
constructors is restricted to exclude the $\strc$ constructor, it is
possible to perform type normalization by only rechecking optimality
of the $\idx$ nodes.  In this case, type normalization can be done in
time proportional to the conciseness of the given tree. When the
$\strc$ constructor is allowed, arbitrarily more concise
representations can be possible as shown in
Figure~\ref{fig:typetreeExample}. Optimality of a subtree that does
not use the $\strc$ constructor does therefore not imply optimality
when $\strc$ is allowed. It is therefore necessary to flatten the
whole tree and apply the tree reconstruction algorithm on the
resulting displacement sequence.

\section{Conclusion}
\label{sec:conclusion}

The main result of this paper is that the type reconstruction problem
into trees is actually solvable in polynomial time. However, an
$O(n^4)$ algorithm is not useful for larger values of $n$ as might be
the case in parallel applications where $n$ could be proportional to
the number of processors which in itself could be in the range of tens
to hundreds of thousands. We note that our bottom-up dynamic
programming algorithm performs a considerable amount of almost
redundant checking for (strided) repetitions in displacement sequence
segments. An asymptotically more efficient algorithm, perhaps based on
a top-down approach, is likely to exist. Whether an exact, practically
efficient algorithm for the full problem is possible, we do not know
at the point of writing.

Restricting the power of the constructors can permit more efficient
algorithms.  As shown in~\cite{Traff14:normalization}, if only
$\leaf,\vct$ and $\idx$ nodes are allowed, then the type
reconstruction problem for a displacement sequence of length $n$ can
be solved in $O(n\sqrt{n})$ time. However, the resulting restricted
trees can and often will be much more costly, as shown in
Figure~\ref{fig:typetreeExample}. The high complexity of our algorithm
is caused by the unbounded branching constructor $\strc$ node. A
slightly better, $O(n^3)$ time algorithm would result from allowing
only bounded branching, for instance a binary struct constructor that
catenates only two subtrees. For such a constructor, the shortest path
computation of Lemma~\ref{lemma:strc} could be done in linear time. In
some contexts, bounded branching might be sufficiently expressive.

An alternative approach would be to look for low-complexity
approximation algorithms with provable approximation guarantees. Or,
even weaker, for heuristics that perhaps work well for the intended
application cases. This reflects the state in current MPI libraries.

As discussed, type trees can be represented more concisely as directed
acyclic graphs (DAGs). To the best of our knowledge, it is still open
whether a cost-optimal DAG representation for an arbitrary
displacement sequence can likewise be constructed in polynomial time.

A related problem to consider is the following. Given two displacement
sequences of the same length, construct a least-cost tree (or DAG)
representing a mapping between the two sequences. Such a tree (DAG)
has uses when copying between different data layouts; this arises,
e.g., in matrix transposition. In the MPI context this operation has
been called
\emph{transpacking}~\cite{Traff08:iotypes,RossLathamGroppLuskThakur09}. Our
dynamic programming algorithm may extend to this case as well.

Our work was specifically inspired by the derived data type mechanism
of MPI. We believe that this idea is applicable in a much wider
context of (parallel) programming interfaces and languages, and that
the type normalization and reconstruction problems as defined here, as
well as the associated processing of data layouts represented by
trees, have relevance extending beyond the motivating context.

\bibliographystyle{plain}
\bibliography{traff,parallel}

\end{document}